\documentclass[letterpaper,11pt]{article}
\usepackage{xcolor}

\def\ShowComment{False} 

\usepackage[english]{babel}

\usepackage{booktabs}
\usepackage{multirow}

\usepackage[numbers, sort]{natbib}

\usepackage[margin=1in]{geometry}

\usepackage{tcolorbox}

\usepackage{fullpage}
\usepackage[utf8]{inputenc}
\usepackage{amsmath,amsfonts,amsthm,amssymb}
\usepackage{setspace}
\usepackage{fancyhdr}
\usepackage{lastpage}
\usepackage{extramarks}
\usepackage{chngpage}
\usepackage{soul,color}
\usepackage{graphicx,float,wrapfig,lipsum}
\usepackage{listings}
\usepackage{xcolor}      
\usepackage{enumerate}
\usepackage{enumitem}
\usepackage{xspace}
\usepackage[T1]{fontenc}

\usepackage{xcolor}
\usepackage{nameref}
\definecolor{ForestGreen}{rgb}{0.1333,0.5451,0.1333}
\definecolor{DarkRed}{rgb}{0.65,0,0}
\definecolor{Red}{rgb}{1,0,0}
\usepackage[linktocpage=true,
pagebackref=true,colorlinks,
linkcolor=DarkRed,citecolor=ForestGreen,
bookmarks,bookmarksopen,bookmarksnumbered]
{hyperref}
\usepackage{cleveref}

\usepackage{indentfirst}
\usepackage{algpseudocode}
\usepackage{algorithm}  
\usepackage{algorithmicx} 
\usepackage{verbatim}
\usepackage{todonotes}
\usepackage{appendix}

\newcommand{\x}{\mathbf{x}}
\newcommand{\y}{\mathbf{y}}

\newcommand{\A}{\mathcal{A}}
\newcommand{\U}{\mathcal{U}}
\newcommand{\V}{\mathcal{V}}
\newcommand{\D}{\mathcal{D}}
\newcommand{\F}{\mathbf{F}}
\newcommand{\E}{\mathbf{E}}
\newcommand{\R}{\mathbb{R}}
\newcommand{\Z}{\mathbb{Z}}

\newcommand{\M}{\mathcal{M}}
\newcommand{\ses}{\mathcal{S}}
\newcommand{\dist}{\mathrm{dist}}

\newcommand{\C}{\mathcal{C}}

\newcommand{\poly}{\mathrm{poly}}

\DeclareMathOperator{\congestion}{cong}

\newcommand{\backvec}[1]{\reflectbox{$\vec{\reflectbox{\!$#1$}}$}}

\makeatletter
\def\namedlabel#1#2{\begingroup
    #2%
    \def\@currentlabel{#2}%
    \phantomsection\label{#1}\endgroup
}
\makeatother

\usepackage{thmtools}

\declaretheorem[numberwithin=section]{theorem}
\declaretheorem[numberlike=theorem]{lemma}

\declaretheorem[numberlike=theorem]{fact}
\declaretheorem[numberlike=theorem]{proposition}

\declaretheorem[numberlike=theorem]{claim}

\declaretheorem[numberlike=theorem,style=definition]{definition}

\usepackage{mathrsfs}  

\Crefname{ALC@unique}{Line}{Lines}

\ifdefined\ShowComment

\def\Mark#1{\marginpar{$\leftarrow$\fbox{T}}\footnote{$\Rightarrow$~{\sf\textcolor{purple}{#1 --Mark}}}}

\def\zhongtian#1{\marginpar{$\leftarrow$\fbox{Z}}\footnote{$\Rightarrow$~{\sf\textcolor{blue}{#1 --Zhongtian}}}}

\else
\def\Mark#1{}
\def\zhongtian#1{}

\fi

\title{Undirected Multicast Network Coding Gaps \\via Locally Decodable Codes}

\author{
Mark Braverman\thanks{Princeton University, Research supported in part by the NSF Alan T. Waterman Award, Grant No. 1933331.}
\and
Zhongtian He\thanks{Princeton University}
}

\date{}

\begin{document}

\begin{titlepage}
    \thispagestyle{empty}
    \maketitle
    \begin{abstract}
        \thispagestyle{empty}
        The network coding problem asks whether data throughput in a network can be increased using coding (compared to treating bits as commodities in a flow). While it is well-known that a network coding advantage exists in directed graphs, the situation in undirected graphs is much less understood -- in particular, despite significant effort, it is not even known whether network coding is helpful at all for unicast sessions. 

In this paper we study the multi-source multicast network coding problem in {\em undirected} graphs. There are $k$ sources broadcasting each to a subset of nodes in a graph of size $n$. The corresponding combinatorial problem is a version of the  Steiner tree packing problem, and the network coding question asks whether the multicast coding rate exceeds the tree-packing rate. 

We give the first super-constant bound to this problem, demonstrating an example with a coding advantage of $\Omega(\log k)$. In terms of graph size, we obtain a lower bound of $2^{\tilde{\Omega}(\sqrt{\log \log n})}$. We also obtain an upper bound of $O(\log n)$ on the gap. 

Our main technical contribution is a new reduction that converts locally-decodable codes in the low-error regime into multicast coding instances. This gives rise to a new family of explicitly constructed graphs, which may have other applications. 
    \end{abstract}
\end{titlepage}

\setcounter{tocdepth}{2}

\renewcommand{\baselinestretch}{0.7}\normalsize
\tableofcontents
\renewcommand{\baselinestretch}{1.0}\normalsize
\thispagestyle{empty}
\newpage
\addtocontents{toc}{\protect\thispagestyle{empty}} 
\setcounter{page}{1}

\section{Introduction}

Optimizing network throughput is a key problem in both combinatorial optimization and coding theory. Relevant objectives include maximizing unicast and broadcast throughput, improving error resilience, enhancing security, etc. Traditional routing treats data packets as physical commodities, giving rise to combinatorial problems such as max-flow and multi-commodity-flow. 
Network coding extends the routing model by allowing data to be encoded and combined at intermediate nodes, potentially offering significant advantages in various settings. While the benefits of network coding in directed graphs are well understood (e.g., \cite{ahlswede2000network,li2003linear,ho2006random,Yeung2006NetworkCoding,fragouli2007network,ho2008network}), key questions in the undirected setting remain open despite much effort.

In fact, a well-known conjecture\footnote{The conjecture is known by several names, including the Network Coding Conjecture, Li and Li's Conjecture, and the Multiple Unicast Conjecture.} states that network coding offers \emph{no} advantage in throughput for \emph{multi-source unicast} in undirected graphs \cite{li2004network,harvey2004comparing}, meaning that the best achievable throughput can be obtained solely through multi-commodity flow routing. Not only the problem is interesting on its own, but it also has strong implications in complexity theory, where the positive resolution of the conjecture would imply lower bounds in external memory algorithm complexity \cite{farhadi2019lower}, in cell-probe model \cite{adler2006capacity}, and even super-linear circuit lower bounds for very natural mathematical problems such as integer multiplication \cite{afshani2019lower}.  Given these connections, it is not surprising that despite numerous attempts \cite{li2004coding,harvey2006capacity,kramer2006edge,langberg2009multiple,xiahou2014geometric,braverman2017coding,yin2017reduction,haeupler2020network} and resolutions on special classes of instances \cite{okamura1981multicommodity,adler2006capacity,jain2006capacity,kramer2006edge}, the conjecture remains open. In contrast, the advantage is known to be $\Omega(n)$ for multi-source unicast in directed graphs of size $n$ \cite{li2004coding,harvey2006capacity}.

Another relevant setting is \emph{single-source multicast}, where one source transmits information to multiple-destinations. In undirected graphs, the known network coding advantage, also called the \emph{coding gap}, lies between $8/7$ \cite{agarwal2004advantage} and $2$ \cite{li2004network}. The coding gap is defined as the ratio of the network coding throughput to the \emph{Steiner tree packing number}, where the latter serves as the non-coding multicast benchmark. We emphasize that a Steiner tree packing combinatorially implies a valid network coding solution: given such a packing, one can transmit information along the Steiner trees, since once a vertex receives some information, it is allowed to forward it to multiple successors.
In directed graphs, the single-source multicast problem is studied in the seminal work on network coding \cite{ahlswede2000network}, which shows that an exact optimal network coding solution can always be achieved. Moreover, the coding gap can be as large as $\Omega(\log n)$ \cite{agarwal2004advantage,jaggi2005polynomial}, matching the integrality gap of linear programming relaxation for the directed Steiner tree problem.

Thereby a natural problem {\em multi-source multicast} generalizing both settings above arises. In this setting, we have multiple source nodes, each communicating its information to a set of destination nodes. The combinatorial non-coding throughput of this problem is \emph{multi Steiner tree packing number}. We study the ratio of multi-source multicast network coding throughput to the corresponding multi Steiner tree packing number in undirected graphs:

{\bf \em ``How large the advantage can network coding obtain for multi-source multicast throughput?''}

Even though the multi-source multicast problem is well-studied, no general asymptotic bounds for the network coding gap in undirected graphs were previously known. 
In a seminal work, Ahlswede et al. \cite{ahlswede2000network} study the multi-source multicast problem in directed graphs, establishing an optimal network throughput for a special case where all sources share the same set of sinks, while leaving the general case as an open problem. Since then, the multi-source multicast problem has been studied extensively from both theoretical and empirical perspectives. To name a few, this includes an exact algebraic formulation of feasibility \cite{koetter2003algebraic}, cut-set outer bounds (e.g., \cite{cover1999elements,harvey2006capacity,thakor2016cut}), extensions to noisy channels \cite{lim2011noisy}, and settings with security guarantees \cite{cohen2018secure}. Notably, Langberg and Médard \cite{langberg2009multiple} show that, for uniform demands, the multi-source unicast non-coding throughput is at least one-third of the multi-source multicast coding throughput. In this context, the multicast problem is defined by modifying each source’s sink set to include all sinks from the unicast problem. This connection positions the multi-source multicast problem in undirected graphs as a partial step toward either proving or disproving the network coding conjecture.


In this work, we prove asymptotic lower and upper bounds for the multi-source multicast problem in undirected graphs. Let $n$ denote the graph size and $k$ the number of sources. We show the existence of network coding instances with a coding gap of $\Omega(\log k)$ for $k$-source multicast sessions, which is the first super-constant bound to this problem. To achieve these coding gaps, we develop novel connections between network coding and locally decodable codes (LDCs), which constitutes the main technical contribution of this work.

In terms of graph size $n$, this coding gap translates into $\displaystyle{2^{\Tilde{\Omega}(\sqrt{\log\log n})}}$ \footnote{We use the notation $\Tilde{\Omega}(f(n))$ to denote $\Omega\left(f(n) / \poly\log f(n)\right)$, i.e., hiding $\poly\log f(n)$ factors.}, utilizing a particular family of LDCs known as matching vector codes \cite{efremenko20093,dvir2011matching,yekhanin2012locally,bhowmick2013new,dvir2013matching}, which excel in low-query regimes.
Since our reduction is black-box, any improvement in LDC constructions would directly lead to stronger lower bounds on the network coding gap. We elaborate on this connection in the technique overview section.

\begin{theorem}
\label{thm: coding gap}
    There exists a family of network coding instances for $k$ multicast sessions in undirected graphs, that have coding gaps at least $\Omega(\log k)$. In term of graph size $n$, the gap is at least $2^{\Omega\left(\sqrt{\log\log n/\log\log\log n}\right)}$.
\end{theorem}


For the upper bound, we show that the coding gap for multi-source multicast sessions is at most $O(\log n)$. Our approach relates the network coding gap to the LP integrality gap for the dual of the multi Steiner packing problem. This integrality gap can, in turn, be bounded using an $O(\log n)$-approximate cut-tree packing \cite{racke2008optimal}, a technique that has also been used in the context of oblivious routing.

\begin{theorem}
\label{thm: upper bound coding gap}
    The multi-source multicast coding gap in undirected graphs is at most $O(\log n)$.
\end{theorem}

This matches the best known upper bound for the coding gap of multi-source unicast in the undirected graph. It is worth noting that, any improvement of the upper bound to $o(\log n)$ even in the multi-source unicast setting would yield a super-linear circuit lower bound for integer multiplication \cite{afshani2019lower}. 

Unlike graph size~$n$, in terms of the number of sources~$k$, the best known upper bound remains the trivial~$O(k)$. 
Closing this exponential gap is an intriguing open problem. In contrast, for multi-source unicast, an upper bound of $O(\log k)$ on the network coding gap is known, as a consequence of the approximate max-flow min-cut theorem (also known as the multicommodity flow-cut gap) ~\cite{aumann1998log,linial1995geometry}.

The observations and results of this work raise several compelling open problems, as the multi-source multicast setting admits a variety of approaches from both coding theory and graph theory. In particular, improving the lower bound in terms of the graph size $n$ may be possible through the construction of more efficient locally decodable codes (LDCs)—especially via a novel variant we introduce, called robust distance LDC codes, discussed in \Cref{sec: robust dist ldc}. Closing the exponential gap in terms of the number of sources $k$ likely requires new combinatorial insights in graph theory. Beyond being interesting on their own, progress on these questions may also serve as intermediate steps toward resolving the network coding conjecture.


\subsection{Related Works}

\label{sec: related}

\paragraph{Completion Time} While it remains unknown whether network coding provides an advantage for the throughput of the multi-source unicast problem in undirected graphs, coding gaps do arise when considering completion time, also known as makespan, instead of throughput. The makespan problem is a generalization of throughput, as the maximum throughput can be defined as $\sup_{r\to\infty}r/C(r)$, where $C(w)$ is the makespan for the instance after increasing all demands by a factor of $r$. For the makespan of multi-source unicast in undirected graphs, Haeupler et al. \cite{haeupler2020network} establish a polylogarithmic coding gap in terms of the number of sources $k$. Additionally, they prove an upper bound that is polylogarithmic in $k$ and the ratio of the summation to the smallest demand for $k$-source unicast instances. It is an interesting open problem to bound the makespan coding gap for multi-source multicast sessions. The completion time of more general distributed computing tasks with security guarantees can also be improved using network coding~\cite{hitron2022broadcast,parter2023secure}.

\paragraph{Graph Product} Previous lower bounds on the network coding gap in undirected graphs are primarily obtained via the graph product technique—a powerful method that amplifies a constant gap into an asymptotically large one. This approach was used by Braverman et al.~\cite{braverman2017coding} to establish a dichotomy in multi-source unicast throughput \cite{braverman2017coding}: the network coding gap for multi-source unicast thoughput is either exactly~1 or at least $\Omega((\log n)^c)$ for some constant $c > 0$. Haeupler et al.~\cite{haeupler2020network} later generalized this framework to show a polylogarithmic coding gap for completion time. While the graph product technique is very useful for establishing lower bounds for undirected multi-source unicast, its applicability to the undirected multicast setting remains unclear. In directed graphs, graph product has been applied to prove polynomial separations between linear and non-linear codes in network coding~\cite{blasiak2011lexicographic,lovett2014linear}.

\paragraph{Locally Recoverable Codes and Fountain Codes} 
Our use of LDCs to improve network coding throughput naturally connects to the broader landscape of fault-tolerant data access in distributed systems. In particular, locally recoverable codes (LRCs)\footnote{Also known as locally reconstructible or locally repairable codes.} have been extensively studied for their role in reducing repair costs in distributed storage settings (e.g., \cite{huang2012erasure,tamo2014family,papailiopoulos2014locally,cadambe2015bounds,barg2017locally,guruswami2019long,martinez2019universal,micheli2019constructions,cai2021construction,kadekodi2023practical}). LRCs aim to minimize the number of nodes accessed during data recovery—an idea closely related to our goal of minimizing the number of queries in LDC-based constructions. Similar motivations underlie work on fountain codes which addressed efficient data dissemination under erasure \cite{byers1998digital,luby2002lt,shokrollahi2006raptor}.


\subsection{Preliminaries}
\label{sec: preliminaries}

\paragraph{Network Coding} A \emph{multi-source multicast instance} $\M = (G, \mathcal{S})$ is defined over a communication network represented by an undirected graph $G = (V, E)$, with edge capacities given by a function $c: E \to \R_+$. The instance consists of $k$ \emph{sessions}, denoted by $\mathcal{S} = \{(s_i, R_i, d_i)\}_{i=1}^k$, where each session is specified by a source vertex $s_i \in V$, a set of sink vertices $R_i \subseteq V$, and a communication demand $d_i \in \R_+$. Let $M_i$ denote the set of all messages that source $s_i$ wishes to send, and define $M := \prod_{i=1}^k M_i$ as the set of all message combinations across sessions. Each message $m_i \in M_i$ is sampled independently from distribution $\D_i$ to produce a distribution $\D$ of $M$'s. 

The formal definition of a \emph{network coding solution} is deferred to \Cref{app: network coding}. Here, we provide a simplified version that, while restricting the full expressive power of network coding, is sufficient for deriving our lower bounds. Given an undirected graph $G$, the solution specifies a direction for each edge $e$ to form a directed graph $\hat{G}$, an associated alphabet $\Gamma(e)$, and an encoding function $f_e : M \to \Gamma(e)$ that determines the symbol transmitted along edge $e$. These functions must satisfy the following conditions:

\begin{itemize}
    \item \textbf{Correctness:} Each sink must be able to recover the message from its corresponding source.
    \item \textbf{Causality:} The symbol transmitted on each edge $e$ must be computable from the symbols received at its tail vertex.
\end{itemize}

We assume that $\hat{G}$ is a directed acyclic graph (DAG). Although the general case may involve cycles, the full definition—based on time-expanded graphs—is presented in \Cref{app: network coding}. That general framework is only required when we analyze upper bounds on network coding throughput in \Cref{sec: ub}. The \emph{throughput} of a network coding solution is defined as the largest value $r$ for which there exists a \emph{scale} $c^* \ge 0$ (i.e., allowing block length, time, or packet size expansion) such that the following conditions hold:

\begin{itemize}
    \item For each source $s_i$, the entropy of its message satisfies $H(m_i) \ge r \cdot d_i \cdot c^*$. 
    \item For each edge $e \in E$, the entropy of the information transmitted along $e$ is at most $c_e \cdot c^*$.
\end{itemize}

The \emph{throughput} of a multi-source multicast instance is the throughput of the network coding solution maximized over all possible distributions $\D$. Note that the common scale $c^*$ in the throughput definition is important: in the network coding literature, throughput characterizes the asymptotic feasibility of the solution.

\paragraph{Multi Steiner Tree Packing} The non-coding throughput of a multi-source multicast problem is captured by the \emph{multi-Steiner packing number}. It is defined as the largest value $\tau$ such that, for each source $s_i$, there exists a Steiner tree packing with terminal set $S_i := R_i \cup \{s_i\}$ and total weight $\tau \cdot d_i$, subject to the constraint that the combined packing across all sources respects the edge capacities. This quantity characterizes the maximum achievable throughput in the absence of network coding—that is, when information can be duplicated but not encoded. A formal LP formulation is given in \Cref{app: Steiner packing}.

\paragraph{Locally Decodable Codes} Denote by $\Delta(\x,\y)$ the Hamming distance between vectors $\x$ and $\y$. Denote by $\A^\y$ an algorithm that accesses the vector $\y$ through coordinate queries. Let $\F$ be a finite field.
A code $\C : \F^k \to \F^N$ is said to be $(q, \delta, \epsilon)$-locally decodable if there exists a randomized decoding algorithm $\A$ such that:

\begin{enumerate}
    \item For all $\x \in \F^k, i\in [k]$ and all $\y \in \F^N$ satisfying $\Delta(\C(\x), \y) \le \delta$, we have
    \[
        \Pr[\A^\y(i) = \x_i] \ge 1 - \epsilon,
    \]
    where the probability is over the internal randomness of $\A$.
    \item The algorithm $\A$ makes at most $q$ queries to $\y$.
\end{enumerate}

The \emph{rate} of a code is defined as the ratio of the message length to the codeword length, i.e., $k/N$. The parameter $q$, known as the \emph{query complexity}, denotes the number of codeword symbols accessed during decoding. The parameter $\delta$ is referred to as the \emph{decoding radius}—the maximum fraction of errors the decoder can tolerate while still recovering individual message symbols with high probability. We say a code satisfies the \emph{linearity of decoding} property (or simply that the code is \emph{linear}) if the decoder $\mathcal{A}$ always outputs a linear combination of the queried codeword symbols.

\begin{definition}[Perfect Smoothness]
    A code satisfies the \emph{perfect smoothness} if, for any $\x \in \F^k$ and any $i \in [k]$, each query made by $\A(i)$ is uniformly distributed over the coordinate set $[N]$, and the decoder always returns the correct message $\x_i$ when given an uncorrupted codeword, i.e. $\A^{\C(\x)}(i) = \x_i$.
\end{definition}

The following fact will be useful in the construction of matching vector codes.

\begin{fact}[\cite{trevisan2004some}]
\label{fact: smooth code imply ldc}
    Any code that satisfies perfect smoothness and makes $q$ queries is also $(q,\delta,q\delta)$-locally decodable for all $\delta < \frac1{q}$.
\end{fact}

\subsection{Technique Overview}
\label{sec: tech overview}

\paragraph{Lower Bound.} The core technical contribution is a reduction from the network coding gap to locally decodable codes.

\begin{lemma}[Main Lemma]
\label{lem: LDC to NC}
    Suppose there exists a linear $(q,\delta,q\delta)$-locally decodable code $\C:\F^k\rightarrow \F^N$ satisfies perfect smoothness. Then, the undirected network coding gap for multi-source multicast is at least $\Omega((\delta\log k)/q)$, with the graph size in the gap instance given by $n = \Theta(Nk)$.
\end{lemma}

This lemma allows us to establish network coding gaps using basic LDCs. For instance, the Hadamard code is a linear $(2,\frac14,\frac12)$-locally decodable code with perfect smoothness, which immediately implies network coding gaps of $\Omega(\log k)$ and $\Omega(\log\log n)$. To achieve a stronger lower bound in terms of $n$, an LDC with a better rate is required.  

In the sub-logarithmic query regime, the best rate is achieved by matching vector (MV) codes \cite{efremenko20093}, formally stated in \Cref{lem: main mv code}, with the proof deferred to \Cref{app: mv code}. Notably, \Cref{lem: main mv code} differs from previous results in several ways. First, the MV codes in \cite{efremenko20093} were analyzed in the constant query regime using Grolmusz's set systems, which become suboptimal for a larger number of queries. Additionally, the MV codes in \cite{dvir2011matching,yekhanin2012locally} require both the distance $\delta$ and error $\epsilon$ to be constant, which slightly increases the query complexity to $t^{O(t)}$. Here, we provide a version optimized for our setting.

\begin{lemma}[Matching Vector Codes -- optimized for network coding gap constructions]
\label{lem: main mv code}
    For $t,k\ge 2$, there exists a linear $q$-query locally decodable code over $\F = GF(2^b)$ with $q = 2^t,b=\Tilde{O}(t)$ that encodes message of length $k$ to codeword of length
    \[
        \exp\exp\left(O\left((\log k)^{1/t}(\log\log k)^{1-1/t}\cdot t\ln t\right)\right).
    \]
    Moreover, the code is $(q,\delta,q\delta)$-locally decodable for any $\delta < 1/q$, and satisfies perfect smoothness.
\end{lemma}

\begin{proof}[Proof of \Cref{thm: coding gap}]
   By applying \Cref{lem: LDC to NC} and the Hadamard code, which is a $(2,\frac14,\frac12)$-locally decodable code, we immediately obtain the lower bound $\Omega(\log k)$.

    To derive lower bounds in terms of $n$, we set $t = \frac14 \log\log k$ and $\delta = \frac1{4q}$ in \Cref{lem: main mv code}, yielding a $\left((\log k)^{\frac14},\frac14(\log k)^{-\frac14}, \frac14\right)$-locally decodable code. This encodes a message of length $k$ into codewords of length
    \[
        N = \exp\exp(O((\log\log k)^2\log\log\log k)).
    \]
    Rearranging terms, we obtain
    \[
        \log k = 2^{\Omega\left(\sqrt{\log\log N/\log\log\log N}\right)}.
    \]
    
    Using this LDC and \Cref{lem: LDC to NC}, we derive a network coding gap of $\Omega(\sqrt{\log k})$, which translates to $2^{\Omega\left(\sqrt{\log\log n/\log\log\log n}\right)}$ in terms of $n$, since $n$ is polynomially related to $N$, and the relation between $N$ and $k$ follows from our previous derivation.
\end{proof}

In \Cref{sec: robust dist ldc}, we propose a relaxed variant of LDCs tailored for applications in network coding gaps, as our approach does not require their full error-correcting capabilities. We believe this variant will be useful in closing the gap between the lower bound in \Cref{thm: coding gap} and the $O(\log n)$ upper bound.


\paragraph{Upper Bound.} We introduce a generalized notion of the sparsest cut in the multi-Steiner setting, defined as the objective of the integral solution to the dual LP of the multi-Steiner packing problem (see \Cref{app: Steiner packing}). This quantity also serves as an upper bound on the network coding throughput, which we formally prove under the information-theoretic definition of network coding. This connection reduces the problem of bounding the network coding throughput to analyzing the LP integrality gap of the multi-Steiner packing problem. Leveraging the $O(\log n)$-approximate cut-tree packing developed for oblivious routing \cite{racke2008optimal}, we obtain an $O(\log n)$ upper bound on the network coding gap. This result can thus be interpreted more strongly: the gap between network coding throughput and oblivious Steiner packing is at most $O(\log n)$.

\paragraph{Presentation Overview} 
In \Cref{sec: undirected lb}, we demonstrate how to construct coding gap instances using LDCs. This section is divided into two parts: in \Cref{sec: sampling ldcs}, we present a simple sampling-based construction to illustrate how LDCs can be applied to network coding; and in \Cref{sec: boosting}, we prove the main lemma by amplifying the coding gap using random binary trees. Although the sampling construction in \Cref{sec: sampling ldcs} is suboptimal, it has the merits of simplicity and motivates a new formulation of LDCs, which we introduce in \Cref{sec: robust dist ldc}. We briefly discuss the limitation of existing LDCs to network coding in \Cref{sec: ldc limit}. In \Cref{sec: ub}, we establish an upper bound on the network coding gap.


\section{Gap Instances via Locally Decodable Codes}

\label{sec: undirected lb}
A \emph{network coding gap instance} $I = (G, \ses, F)$ is a multi-source multicast instance $\M = (G, \ses)$ defined over a connected graph $G$, along with a designated set of cut edges $F \subseteq E(G)$. We restrict our attention to instances where all edge capacities and demands are unit-valued, i.e., $c_e = 1$ for all $e \in E(G)$ and $d_i = 1$ for all $i \in [k]$. Let $\dist_G(u, v)$ denote the shortest-path distance between nodes $u$ and $v$ in $G$. We say that a gap instance has parameters $(a, b, f, m, r)$ if the following conditions hold:

\begin{enumerate}
    \item The $k$-source multicast instance $\M$ admits a network coding solution achieving throughput $a$.
    \item Removing the edges in $F$ disconnects $s_i$ from $R_i$ for every $i \in [k]$.
    \item Let $\dist_{G \setminus F}(U) := \min_{\substack{u,v \in U \\ u \ne v}} \dist_{G \setminus F}(u, v)$. Then for all $i \in [k]$, it holds that $\dist_{G \setminus F}(R_i) \ge b$.
    \item The number of cut edges is $f = |F|$.
    \item The total number of sinks is $r = \sum_{i \in [k]} |R_i|$.
    \item The graph $G$ contains at most $m$ edges, i.e., $|E(G)| \le m$.
\end{enumerate}

\begin{lemma}
\label{lem: gap instance}
Let $I$ be a gap instance with parameters $(a, b, f, m, r)$. Then the network coding gap is at least $\frac{a \cdot r}{f + 2m/b}$.
\end{lemma}

\begin{proof}
Let $\tau$ denote the non-coding throughput, characterized by the multi-Steiner packing number. By the LP duality in \Cref{app: Steiner packing}, the dual objective is an upper bound on $\tau$. Now we shall assign dual variables: let $z_i = |R_i| + 1$ for every demand $i$; $y_e = 1$ for every edge $e\in F$ and $y_e = 2/b$ otherwise. We show that the dual constraints are satisfied. Clearly, all the variables are non-negative. Recall that we define $S_i := R_i\cup \{s_i\}$. For any $i$ and tree $T$ such that $S_i\subseteq V(T)$, we need to show that the constraint $\sum_{e\in T} y_e \ge z_i$ is satisfied. Let $r_T := |R_i|$ be the number of sinks spanned by $T$, and let $x_T$ denote the number of cut edges $F$ used by $T$.

By the definition of a gap instance, in $G \setminus F$, any two sinks in $R_i$ are at distance at least $b$. We say a sink $t \in R_i$ is \emph{close} to a cut edge $e = (u,v) \in T \cap F$ if, in $G \setminus F$, the vertex $u$ is connected to $s_i$, $v$ is not, and $\dist_{G \setminus F}(t, v) < b/2$. By the distance condition of the gap instance, at most one sink in $R_i$ can be close to any such cut edge. Therefore, the number of close sinks in $T$ is at most $x_T$. 

For each one of the remaining sinks $s$, let $N(s,b/2)$ be the edges in $(b/2)$-neighborhood of $s$ in $T\setminus F$. Since $s$ is remote, we have $|N(s,b/2)|\ge b/2$. In addition, for each two distinct $s_1\neq s_2$ we have $\dist_{T\setminus F}(s_1,s_2)\ge\dist_{G\setminus F}(s_1,s_2) \ge b$, and thus $N(s_1,b/2)\cap N(s_2,b/2)=\emptyset$.
 Therefore, the total number of non-cut edges in $T$ is at least
 $(r_T-x_T)\cdot \frac{b}{2}$. Thereby the dual constraint is satisfied, 
 \[
    \sum_{e\in T} y_e \ge \frac{2}{b}\cdot (r_T-x_T) \frac{b}{2} + 1\cdot x_T = r_T = z_i.
 \]

With this assignment, the dual objective is,
\[
\frac{\sum_{e\in E} c_e y_e }{ \sum_i d_i z_i} \le \frac{f + 2m/b}{r}
\]
which upper bounds the packing number $\tau$.
Given that the coding throughput is $a$, the network coding gap is at least
\[
\frac{a}{\tau} \ge \frac{a \cdot r}{f + 2m/b}.
\]
\end{proof}

\subsection{Sampling Construction via LDCs}

\label{sec: sampling ldcs}

Fix a perfectly smooth $(q,\delta,q\delta)$ code $\C : \F^K \to \F^N$, we will show that the following graph construction is a $(a,b,f,m,r)$ gap instance for $a=1, b=\log k/ \log\log k, f=N, m=\Theta(rq)$, and $r=\Theta((\delta N\log k)/q)$. We only consider gap instances where all sources are co-located at the same vertex $S$. Our construction is based on a tripartite graph $G = (A, B, C, E)$, structured as follows:

\begin{itemize}
    \item The first part $A$ contains only the common source vertex $S$, i.e. $s_i = S$ for all $i\in [k]$. In a network coding solution, the message of $s_i$ is represented as $\x_i\in \F$. Thereby, vertex $S$ knows the entire message vector $\x\in \F^k$.
    \item The second part $B = \{w_j\}_{j=1}^N$ consists of $N$ vertices and each is connected to $S$ by an edge. In the network coding solution, each vertex $w_j\in B$ will receive from $S$ one coordinate of the codeword $\C(\x)_j$. 
    \item The third part $C$ is the disjoint union of sink sets: $C = \bigsqcup_{i=1}^k R_i$, where $R_i \cap R_j = \emptyset$ for all $i \ne j$. Each sink in $R_i$ connects to $q$ vertices in $B$ that enable decoding of the message $\x_i$, using the local decodability of the code. In the remainder of this subsection, we describe in detail how to construct the third part.
\end{itemize}

We begin by identifying the property that the third part should satisfy. To apply \Cref{lem: gap instance}, we set the cut edges $F$ to be all edges adjacent to $S$, and aim to ensure that the pairwise distances between sinks in each $R_i$ are not too small in $G \setminus F$. This is achieved by first generating a near-linear number of candidate sinks for each $i$, then retaining only a $(\log k)/k$ fraction via sampling. A pruning subroutine is applied to eliminate sinks that are too close to others associated with the same source. Both subroutines are formalized below. The pseudocode is shown in \Cref{alg: gap instance LDC}.

\paragraph{Sampling Subroutine} For $q$-query LDCs, each vertex in $R_i$ must have $q$ neighbors in the second layer $B$ that can successfully decode the message $\x_i$. The code's decoding radius is used to generate many disjoint $q$-queries of near-linear size. Specifically, we maintain a set of used indices $U \subseteq [N]$. As long as $|U| \le \delta N$, there exists a set $D \subseteq [N] \setminus U$ of size $q$ such that $\x_i$ can be decoded by querying $\C(\x)$ at coordinates in $D$ (see \Cref{fact: ldc decode} and line~\ref{line: decode} in \Cref{alg: gap instance LDC}). Denote by $\Call{Decode}{i, U}$ the oracle that returns the corresponding set~$D$. All indices in $D$ are added to $U$.

With probability $\log k/k$, we add a vertex $v$ to the sinks $R_i$ and add edges from $v$ to the vertices in $B$ indexed by $D$. Note that $D$ is added to $U$ regardless of whether $v$ is sampled. The process terminates when $|U| > \delta N$. The sampling subroutine runs from line~\ref{line: sampling begin} to line~\ref{line: sampling end} in \Cref{alg: gap instance LDC}.

\begin{fact}
\label{fact: ldc decode}
Let $C : \F^k \rightarrow \F^N$ be a $(q, \delta, q\delta)$-locally decodable code with perfect smoothness. Then for any index $i \in [k]$ and any subset $U \subseteq [N]$ with $|U| \le \delta N$, there exists a set $D \subseteq [N] \setminus U$ of size at most $q$ such that the message symbol $\x_i$ can be successfully decoded by querying the codeword $\C(\x)$ at positions in $D$.
\end{fact}

\begin{proof}
By the perfect smoothness property, each of the $q$ queries made by the decoder is uniformly distributed over $[N]$, so the probability that any single query falls in the corrupted set $U$ is at most $\frac{|U|}{N} \le \delta$. By the union bound, the probability that any of the $q$ queries falls in $U$ is at most $q\delta$. Therefore, with probability at least $1 - q\delta$, all $q$ queries lie entirely within the uncorrupted set $\overline{U}$.

By the probabilistic method, this implies the existence of a fixed set $D \subseteq \overline{U}$ of size $q$ such that the decoder queries exactly the positions in $D$ with some fixed randomness. Since all queried positions are uncorrupted and the decoder behaves correctly on valid codewords by the perfect smoothness assumption, $\x_i$ can be deterministically decoded from the values of $\C(\x)$ at positions in $D$.
\end{proof}

\begin{algorithm}[htbp]
\caption{Generating Gap Instance via LDCs}
\label{alg: gap instance LDC}
\begin{algorithmic}[1]
\State Set $A \gets \{S\}$ and $s_i \gets S$ for all $i \in [k]$.
\State Set $B \gets \{w_i\}_{i=1}^N$.
\State Set $F \gets \{(S, w_i) \mid i \in [N]\}$ and initialize $E \gets F$.
\State Initialize $R_i \gets \emptyset$ for all $i \in [k]$.
\For{\textbf{each} $i = 1$ to $k$} \Comment{Sampling Subroutine.} \label{line: sampling begin}
    \State Initialize $U \gets \emptyset$.
    \While{$|U| \le \delta N$}
        \State Let $D \gets \Call{Decode}{i, U}$. \Comment{Select $q$ coordinates in $[N]\setminus U$ for decoding.} \label{line: decode}
        \State \textbf{with probability} $(\log k)/k$:
        \State \hspace{1em} Add a new vertex $v$ to $R_i$.
        \State \hspace{1em} Add edges $(w_j, v)$ to $E$ for each $j \in D$.
        \State Update $U \gets U \cup D$.
    \EndWhile
\EndFor \label{line: sampling end}
\For{\textbf{each} $i \in [k]$ and \textbf{each} $u \in R_i$} \Comment{Pruning Subroutine.} \label{line: pruning begin}
    \If{there exists $v \in R_i \setminus \{u\}$ such that $\dist_{G \setminus F}(u, v) \le \log k / \log\log k$} \label{line: pruning cond}
        \State Remove $v$ from $R_i$.
    \EndIf
\EndFor \label{line: pruning end}
\State Set $C \gets \bigcup_{i=1}^k R_i$.
\end{algorithmic}
\end{algorithm}

\paragraph{Pruning Subroutine}  
We prune the set of sinks $C$ to ensure that, within each $R_i$, the sinks are sufficiently far apart in the graph $G \setminus F$. Specifically, for each $i \in [k]$ and each sink $u \in R_i$, if there exists another sink $v \in R_i$ such that the distance between $u$ and $v$ in $G \setminus F$ is less than $\log k / \log\log k$, we remove $u$ from $R_i$. 

The pruning threshold is chosen to ensure that only a small number of sinks are removed from $C$, as formally shown in \Cref{lem: prune LDC instance}. After pruning, we apply \Cref{lem: gap instance} to establish a lower bound on the coding gap.

\begin{lemma}
    \label{lem: prune LDC instance}
    Assume that $q = o(\log k)$. Then, in the pruning subroutine of \Cref{alg: gap instance LDC}, an $o_k(1)$\footnote{That is, a quantity tending to $0$ as $k \to \infty$.} fraction of the sinks in $C$ is removed with high probability.
\end{lemma}

\begin{proof}
We analyze the pruning process in the graph $G \setminus F$. Fix any $i \in [k]$ and any sink $u \in R_i$. Let $P_j$ be the number of vertices in $B$ within distance $2j + 1$ from $u$. Initially, $P_0 = q$. For $j > 0$, we have the recurrence
\[
\E[P_j] \le q \log k \cdot \E[P_{j-1}],
\]
since each vertex in $C$ has $q$ neighbors in $B$, and each vertex in $B$ has at most $k$ potential neighbors in $C$, with each edge sampled independently with probability $\log k / k$.

Now let $Q_j$ be the number of vertices in $R_i \setminus \{u\}$ within distance $2j$ from $u$. Clearly, $Q_0 = 0$. For $j > 0$, the expectation $\E[Q_j]$ equals $\E[Q_{j-1}]$ plus the expected number of vertices in $R_i \setminus \{u\}$ that are at distance exactly $2j$ from $u$.
The latter quantity can be upper bounded by the expected number of vertices in $B$ at distance $2j - 1$ from $u$, multiplied by the probability that a vertex in $B$ is connected to a vertex in $R_i$.
This probability is at most $\frac{\log k}{k}$ as the sampling procedure, implying
\[
\E[Q_j] \le \E[Q_{j-1}] + \frac{\log k}{k} \cdot \E[P_{j-1}] \le \frac{(q \log k)^{j+1}}{k} ~.
\]

Let $d^* = \frac{1}{2} \cdot \frac{\log k}{\log \log k}$. The algorithm prunes $u$ if there exists another sink in $R_i$ within distance $2d^*$ in $G \setminus F$. Using the bound above, we get
\[
\E[Q_{d^*}] \le \frac{(q \log k)^{d^*}}{k} = o(1),
\]
and therefore
\[
\Pr[Q_{d^*} \ge 1] \le \E[Q_{d^*}] = o(1),
\]
which gives an upper bound on the probability that $u$ is pruned.

Since this holds for every $u \in R_i$ and all $i \in [k]$, the expected fraction of sinks pruned is $o(1)$. Applying Markov's inequality concludes the proof.
\end{proof}

We name this instance $I_S = (G,\ses,F)$. Applying \Cref{lem: gap instance} to our gap instance $I_S$, we obtain the following lemma.
\begin{lemma}
\label{lem: sampling nc gap}
Suppose $\C : \F^k \rightarrow \F^N$ is a $(q, \delta, q\delta)$-locally decodable code with perfect smoothness, as used in \Cref{alg: gap instance LDC}. Then the network coding gap for the resulting instance $I_S$ is at least
\[
\Omega\left( \min\left\{ \delta, \frac{1}{\log\log k} \right\} \cdot \frac{\log k}{q} \right).
\]
\end{lemma}

\begin{proof}
For the network coding throughput, we consider each source message to be uniformly distributed over $\F$. Each edge in the network carries a symbol in $\F$, and each sink decodes its target symbol using $q$ symbols from the second layer. Each vertex in $B$ receives one coordinate of the locally decodable codeword, and each vertex in the third layer accesses $q$ such coordinates.

Each source has entropy $\log s$, where $s = |\F|$. Since all demands are unit-valued, we scale the capacities by $c^* = \log s$, as in the definition of network coding solution. Consequently, the effective capacity of each edge is $1 \cdot c^* = \log s$, since all edge capacities are unit-valued. Under this scaling, the network supports a coding solution that delivers the full message to each sink, and therefore the coding throughput of $I_S$ is at least $a = 1$.

We now analyze the remaining parameters required by \Cref{lem: gap instance} for $I_S$. By the sampling subroutine, the number of sinks $r = \Theta\left( \frac{\delta N}{q} \log k \right)$, which dominates the total number of vertices $n$. The total number of edges is $m = \Theta(rq)$. The pruning subroutine ensures that the minimum pairwise distance among sinks in each $R_i$ is at least $b \ge \log k / \log\log k$, and the number of cut edges is $f = N$.

Applying \Cref{lem: gap instance}, the coding gap is at least
\[
\frac{a \cdot r}{f + 2m / b} = \Omega\left( \min\left\{ \frac{r}{f}, \frac{r}{2m / b} \right\} \right).
\]
Substituting the estimates for $r$, $m$, and $b$ gives:
\[
\Omega\left( \min\left\{ \frac{\delta \log k}{q}, \frac{\log k}{q \log\log k} \right\} \right),
\]
as claimed.
\end{proof}

\subsection{Boosting the Distance via a Binary Tree Gadget}
\label{sec: boosting}

In this section, we introduce a binary tree gadget to eliminate the low-order term in the lower bound on the coding gap. In the previous construction, the distance between sinks in each $R_i$ was ensured through a sampling-based approach, which introduced an extra $\log \log k$ factor. Here, we replace high-degree vertices with structured binary trees to achieve the better distance guarantees without relying on sampling. The pseudocode is provided in \Cref{alg: gap instance LDC BT}.

\paragraph{Binary Tree Gadget}  
This construction can be viewed as a modification of the previous bipartite structure. Denote by $G[B,C]$ the bipartite graph induced by the vertex set $B \cup C$. For each edge $(u, v)$ in $G[B, C]$, we first introduce a new intermediate vertex $t$ and replace the edge with two edges $(u, t)$ and $(v, t)$.

Next, for each vertex $u \in B \cup C$ with degree $d$ in the graph, we remove all edges incident to $u$ and replace them with a complete binary tree rooted at $u$ with $d$ leaves. The original neighbors of $u$ are then assigned to the leaves via a uniformly random permutation—each leaf is connected to a unique neighbor of $u$ according to this random ordering.

After this transformation, every intermediate vertex $t$ (originally inserted to split edges) has degree exactly $2$, so we remove $t$ and directly connect its two neighbors with an edge. Although the construction in \Cref{alg: gap instance LDC BT} presents this process in a more direct manner, it is equivalent to the description above.

Notably, the sampling subroutine from the previous section is no longer required in this construction.

\begin{algorithm}[h]
\caption{Generating Gap Instance via LDCs and Binary Tree Gadget}
\label{alg: gap instance LDC BT}
\begin{algorithmic}[1]
\State Set $A \gets \{S\}$ and $s_i \gets S$ for all $i \in [k]$.
\State Set $B \gets \{w_i\}_{i=1}^N$.
\State Set $F \gets \{(S, w_i)\}_{i=1}^N$ and initialize $E \gets F$.
\State Initialize $R_i \gets \emptyset$ for all $i \in [k]$.
\For{\textbf{each vertex} $u \in B$} \Comment{Insert binary tree for each $B$-vertex}
    \State Insert a complete binary tree $T_u$ with $k$ leaves and root at vertex $u$.
    \State Label the leaves of $T_u$ with a random permutation of $\{1, 2, \ldots, k\}$.
\EndFor
\For{\textbf{each} $i = 1$ to $k$} \Comment{Add sinks via LDC decoding}
    \State Initialize $U \gets \emptyset$.
    \While{$|U| \le \delta N$}
        \State Let $D \gets \Call{Decode}{i, U}$.
        \State Add a new vertex $v$ to $R_i$.
        \State Insert a complete binary tree $T_v$ with $q$ leaves and root at vertex $v$.
        \State Let $p: D \to [q]$ be an arbitrary bijection (permutation of $D$).
        \For{\textbf{each} $j \in D$}
            \State Add an edge from the leaf of $T_{w_j}$ labeled $i$ to the $p(j)$-th leaf of $T_v$.
        \EndFor
        \State Update $U \gets U \cup D$.
    \EndWhile
\EndFor
\For{\textbf{each} $i \in [k]$ and \textbf{each} $u \in R_i$} \Comment{Prune close sinks}
    \If{there exists $v \in R_i \setminus \{u\}$ such that $\dist_{G \setminus F}(u, v) \le \frac{1}{4} \log k$}
        \State Remove $v$ from $R_i$.
    \EndIf
\EndFor
\end{algorithmic}
\end{algorithm}

\paragraph{Pruning Subroutine}  
Similar to the previous section, the pruning threshold is chosen to ensure that only a small fraction of sinks are removed. The key advantage of the current construction is that the minimum required distance between sinks can now be set to a logarithmic function of $k$. 

Intuitively, the analysis relies on the structure of the binary tree gadget. Starting from any vertex in a tree and growing a ball of increasing radius, the next labeled vertex encountered is essentially random due to the uniform permutation of labels. Thus, to encounter a vertex with the same label (i.e., corresponding to the same message index $i$), the ball must typically grow to logarithmic radius. This ensures that sinks associated with the same message are, with high probability, far apart in $G \setminus F$.

\begin{lemma}
\label{lem: prune LDC instance BT}
In the pruning subroutine of \Cref{alg: gap instance LDC BT}, at most an $O(k^{-1/2})$ fraction of the sinks are removed with high probability.
\end{lemma}

The correctness of \Cref{lem: prune LDC instance BT} follows from the claim below.
\begin{claim}
\label{clm: prune BT}
Before the pruning subroutine, for each $i \in [k]$ and each $u \in R_i$, with probability at least $1 - k^{-3/4}$,
\[
\min_{v \in R_i \setminus \{u\}} \dist_{G \setminus F}(u, v) \ge \frac{1}{4} \log k.
\]
\end{claim}

\begin{proof}
    The only randomness in the algorithm arises from labeling the leaves of each binary tree using independent random permutations. For the sake of analysis, we may equivalently assume that all sinks are added first, and then the leaf labels of every binary tree are assigned uniformly at random. 

     Let $D$ be the indices we used to create sink $u$. Consider a path from $u$ to some $v\in R_i\setminus \{u\}$. The path should cross some binary tree $T_{w_j}$ for some $j\not\in D$. This motivates us to define the following set
     \[
        P = \{\, v \in G \setminus S : \dist_{G \setminus F}(u, v) < \tfrac{1}{4}\log k \,\} 
        \setminus \bigl( \cup_{j \in D} V(T_{w_j}) \cup V(T_u) \bigr).
    \]
     Since the degree of each vertex is at most 3 in $G\setminus F$ and the degree of $u$ is 2, we have $s:=|P| < k^{\frac14}$.

    Let $C = \bigcup_{i\in [k]} R_i$. Let $B' = \bigcup_{w\in B} V(T_w)$ and $C' = \bigcup_{w\in C} V(T_w)$. For each vertex $v \in P$, we say that $v$ is bad if either $v \in B'$ is labeled by~$i$, or $v \in C'$ lies in $T_w$ for some $w \in R_i$. If there is no bad vertex in $P$, by the definition of $P$ and the discussion above, the inequality in the claim holds.
    
    For the vertices in $P$, we order them by their distance to the vertex $u$ in $G\setminus F$, i.e. $P = \{v_1,v_2,\ldots,v_s\}$ such that $\dist_{G\setminus F}(u,v_\ell)\le \dist_{G\setminus F}(u,v_{\ell+1})$, then we shall prove that
    \begin{align}
    \label{eqn: ub bad prob}
        \Pr\bigl[v_\ell\text{ is bad}\mid v_j\text{ is good } \forall j<\ell \bigr] \le \frac{1}{k-\ell}
    \end{align}
    which will imply that,
    \begin{align*}
        \Pr\bigl[\forall v\in P, v \text{ is good}\bigr] &\ge \prod_{\ell=1}^s \Pr\bigl[v_\ell\text{ is good}\mid v_j\text{ is good } \forall j<\ell\bigr]\ge \prod_{\ell=1}^s \frac{k-\ell-1}{k-\ell}\ge 1 - k^{-\frac34}.
    \end{align*}
    This will conclude the claim as we discussed in the beginning. The rest is proving \Cref{eqn: ub bad prob}. Let $d = \dist_{G\setminus F}(u,v_\ell)$. There are two cases.
    \begin{enumerate}
        \item There exists a vertex $v^*\in C'$ adjacent to $v_\ell$, and $\dist_{G\setminus F}(v^*,u)<d$. Let $w$ be the sink vertex such that $v^*\in T_w$. By assumption, $w\not\in R_i$ as $v^*$ is good, thereby  $v_\ell$ is not labeled by $i$, nor in another binary tree $T_{w'}$, which means $v_\ell$ is good.
        \item Otherwise, there exists a vertex $v^*\in B'$ adjacent to $v_\ell$, and $\dist_{G\setminus F}(v^*,u)<d$. We only need to consider the case when $v_\ell$ is a leaf of the binary tree, in which $v_\ell$ has a label. Let $T_{w_j}$ be the tree contains $v^*$. If $j\in D$ then the label of $v^*$ must be different from $i$, since otherwise $v_\ell$ should be in $T_u$ which contradicts the definition of $P$. Now $j \notin D$. Among the first $\ell$ vertices in $P$, there are at least $k - \ell$ unseen labels, including~$i$.
        Since the label of $v_\ell$ is assigned uniformly among these remaining labels as in a random permutation, the probability that $v_\ell$ receives label~$i$ is at most $\frac{1}{k - \ell}$. \qedhere
    \end{enumerate}
\end{proof}
\begin{proof}[Proof of \Cref{lem: prune LDC instance BT}]
    By \Cref{clm: prune BT}, $k^{-\frac34}$ is an upper bound of the probability that each sink $u$ be pruned. Therefore, the expected fraction of the set pruned is $k^{-\frac34}$. Then we conclude with Markov's inequality.
\end{proof}

We name this instance $I^* = (G,\ses,F)$. Applying \Cref{lem: gap instance} to our gap instance $I^*$, we prove our main lemma.

\begin{lemma}[Restatement of \Cref{lem: LDC to NC}]
\label{lem: LDC to NC restated}
Suppose there exists a linear $(q, \delta, q\delta)$-locally decodable code $\C : \F^k \rightarrow \F^N$ with perfect smoothness. Then the network coding gap for the instance $I^*$ is at least $\Omega\left( \frac{\delta \log k}{q} \right)$, where the graph has size $n = \Theta(Nk)$.
\end{lemma}

\begin{proof}
The proof follows the same outline as that of \Cref{lem: sampling nc gap}. For the network coding throughput, we assume each source message is uniformly distributed over $\F$. Each edge in the network carries a symbol from $\F$, and each sink recovers its target symbol using $q$ symbols from the second layer. Specifically, each vertex in $B$ receives one coordinate of the codeword generated by the locally decodable code, and each sink in the third layer wants to query $q$ such coordinates.

The binary tree $T_v$ rooted at each $v \in B$ propagates the value of the corresponding codeword coordinate toward its leaves. For each sink $r$, the binary tree $T_r$ aggregates the $q$ codeword symbols received from the second layer and computes the linear decoding function as the data flows upward toward the root. In this way, the network performs the entire decoding process within its structure.

The scale defined in the network coding solution is set by $c^* = \log s$, which is the same as in \Cref{lem: sampling nc gap}. Under this scaling, the network supports a coding solution that delivers the full message to each sink, and therefore the coding throughput of $I_S$ is at least $a = 1$.

We now estimate the parameters in \Cref{lem: gap instance}. For each $i \in [N]$, the algorithm constructs a binary tree of size $\Theta(k)$ rooted at $c_i \in B$, contributing $\Theta(Nk)$ vertices. Additionally, each sink is augmented with a binary tree of size $\Theta(q)$. The total number of sinks is $\Theta(\delta Nk / q)$, so the overall graph size is $n = \Theta(Nk)$ and the number of edges is $m = \Theta(n)$, as all vertices have constant degree except for the source $S$.

The number of sinks that remain after pruning is $r = \Theta(\delta Nk / q)$ by \Cref{lem: prune LDC instance BT}. The distance parameter is $b > \frac{1}{4} \log k$, as ensured by the pruning subroutine, and the number of cut edges is $f = N$.

Applying \Cref{lem: gap instance}, the coding gap is at least:
\[
\frac{r}{f + 2m / b} = \Omega\left( \min\left\{ \frac{r}{f}, \frac{r b}{2m} \right\} \right) = \Omega\left( \min\left\{ \frac{\delta k}{q}, \frac{\delta \log k}{q} \right\} \right) = \Omega\left( \frac{\delta \log k}{q} \right).
\]
\end{proof}

\section{Upper Bound Network Coding Gap}

\label{sec: ub}

Given a graph $G = (V,E)$ with capacities $c_e$ for each edge $e\in E$. It will be convenient to define $c_{uv} = c_e$ for all edge $e = (u,v)\in E$ and $c_{uv} = 0$ otherwise. For any set of vertex $U\subseteq V$, we define the cut capacity $C(U,\overline{U}) = \sum_{u\in U,v\in \overline{U}} c_{uv}$. Recall that $S_i = R_i\cup \{s_i\}$ is the set of terminals including source and all sinks for demand $i$. Then, for any set of vertex $U$, we define the demand crossing $U$ to be
\[
    D(U,\bar{U}) := \sum_{\substack{1\le j\le k\\ 0<|S_j\cap U|<|S_j|}} d_i
\]

Define $\Psi = \min_{U} \frac{C(U,\overline{U})}{D(U,\overline{U})}$ which generalizes the definition of the sparsest cut. It is straightforward to verify that the value of the integral solution to the dual LP in \Cref{app: Steiner packing} is at most~$\Psi$, since the dual objective achieves~$\Psi$ under the assignment $y_e = z_i = 1$ for all edges~$e$ and demands~$i$ crossing~$U$, and $y_e = z_i = 0$ otherwise. Furthermore, we show that $\Psi$ is an
upper bound on the network coding throughput.

\begin{lemma}
\label{lem: general sparsity ub coding}
    The multi-source multicast network coding throughput is at most $\Psi$.
\end{lemma}

\begin{proof}
    Let $U$ be the vertex set such that $\Psi = \min_{U} \frac{C(U,\overline{U})}{D(U,\overline{U})}$. Suppose there are $\ell$ edges crossing $U$ and $\overline{U}$, denoted by $e_1,\ldots,e_{\ell}$. Recall the definition in \Cref{app: network coding}, each edge is splitted into two directed edge and $\Gamma(\vec{e})$ represent the alphabet on edge $\vec{e}$. Let $X_{\vec{e}}\in \Gamma(\vec{e})$ be the symbol transmitted on edge $\vec{e}$. Let
    \[
        X_l = (X_{\overrightarrow{e_1}}, X_{\overrightarrow{e_2}}, \ldots,X_{\overrightarrow{e_\ell}})
    \]
    be all the symbols transmitted on the edges from $U$ to $\overline{U}$, and
    \[
        X_r = (X_{\overleftarrow{e_1}},X_{\overleftarrow{e_2}},\ldots,X_{\overleftarrow{e_\ell}})
    \]
    Suppose there are $p$ terminal sets $S_{j_1},\ldots,S_{j_p}$ that cross $U$ with source node in $U$, i.e. $s_{j_i}\in U$ and $R_{j_i}\cap \overline{U}\ne \emptyset$, and $q$ terminal sets $S_{j'_1},\ldots,S_{j'_q}$ that cross $U$ with source node in $\overline{U}$, i.e. $s_{j'_i}\in \overline{U}$ and $R_{j'_i}\cap U\ne \emptyset$. Let $m_j\in M_j$ be the message of the $j$-th source. Let
    \[
        Y_l = (m_{j_1},m_{j_2},\ldots,m_{j_p})
    \]
    and
    \[
        Y_r = (m_{j'_1},m_{j'_2},\ldots,m_{j'_q})
    \]
    such that $Y_l$ and $Y_r$ together contains all the messages relevant to this cut $(U,\overline{U})$, and we define $Z$ to be all the other messages. We want to show that
    \[
        H(X_l)\ge H(Y_l)\quad \text{and}\quad H(X_r)\ge H(Y_r),
    \]
    then the lemma follows by definition of network coding throughput. To see this, suppose by contradiction that the throughput exceeds $\Psi$: there exists scale constant $c^*$ such that $H(Y_l\cup Y_r) > D(U,\overline{U})\cdot \Psi\cdot c^*$, we have $H(X_l\cup X_r) > D(U,\overline{U})\cdot \Psi\cdot c^* = C(U,\overline{U})\cdot c^*$, which implies that there exists an edge $e$ such that $H(e) > c_e\cdot c^*$, contradicts to the capacity constraint.
    
    We claim that the conditional entropy $H(Y_l\mid X_l,Y_r,Z) = 0$ and $H(Y_r\mid X_r,Y_l,Z) = 0$. Without loss of generality we prove the first one, we need to show that each message $m_{j_i}$ is determined by $X_l,Y_r$ and $Z$, this is indeed as in the time-expanded graph which is acyclic, the output can be computed once the input is fixed, and once we truncate all the vertices and edges in $U$ which are prior to $X_l$, every symbol on all the edges could still be computed given $X_l,Y_r$ and $Z$. By assumption, there exists $r\in R_{j_i}$ such that $r\in \overline{U}$, hence $m_{j_i}$ can be computed for all $i$, which implies our claim.

    Furthermore, $Y_l,Y_r$ and $Z$ are mutually independent. Therefore by the claim we derived and applying chain rule, $I(Y_l;X_l\mid Y_r,Z) = H(Y_l\mid Y_r,Z) - H(Y_l\mid X_l,Y_r,Z) = H(Y_l) - 0 = H(Y_l)$. On the other hand, $I(Y_l;X_l\mid Y_r,Z)\le H(X_l\mid Y_r,Z)\le H(X_l)$. Combining together, we have $H(X_l)\ge H(Y_l)$, and same holds for $H(X_r)\ge H(Y_r)$.
\end{proof}

Therefore, the proof of \Cref{thm: upper bound coding gap} is complete once we have the following lemma, which gives an upper bound on the LP integrality gap. The rest of this section is to prove \Cref{lem: main LP int gap}.

\begin{lemma}
\label{lem: main LP int gap}
    The ratio of the generalized sparsity $\Psi$ over the multi Steiner tree number is at most $O(\log n)$.
\end{lemma}

The proof builds upon the results from oblivious routing \cite{racke2008optimal} using the \emph{tree-based flow-sparsifiers}. We first introduce the preliminaries. Then we use the tree-based flow-sparsifiers to lower bound the Steiner tree packing number by a $O(\log n)$-factor of relaxation of the generalized sparsity $\Psi$.

The generalized tree packing of a graph $G$ consists of a collection of trees $T_i$ and distribution $\lambda_i$ (i.e. $\sum_i \lambda_i = 1$) such that each tree spans the vertex set $V$ but not necessarily a subgraph of $G$, along with an embedding of the tree $T_i$ in the graph, i.e. a mapping $P_i$ that maps each edge $(x,y)\in T_i$ to a path in $G$, which denoted by $P_i(x,y)$. For each $(x,y)\in T_i$, removing $(x,y)$ separates $T_i$ into two parts $(U,V\setminus U)$. We define $C_i(x,y) := C(U,\overline{U})$ as the size of the cut, and define $D_i(x,y) := D(U,\overline{U})$ which captures the weighted demand that crosses the edge $(x,y)$. Now we are ready to define the proposition that is required for the tree packing.

\begin{definition}[Cut-Tree Packing]
    It is called an \emph{$\alpha$-approximate cut-tree packing}, if a collection of trees $T_i$ with a distribution $\lambda_i$ satisfies the following inequality for every edge $(u,v)\in E$,

    \[
        c_{uv}\ge \frac{1}{\alpha} \sum_i \lambda_i\sum_{(x,y)\in T_i: (u,v)\in P_i(x,y)}C_i(x,y) ~.
    \]
    i.e. the capacity of $(u,v)$ is at least $1/\alpha$ fraction of the sum of weighted cuts associated with the paths using $(u,v)$.
\end{definition}

Given the terminal set $S_j := \{s_j\}\cup R_j$, we want to define a multi Steiner tree packing solution based on each tree $T_i$ we constructed for the cut-tree packing. For each tree $T_i$, there exists a minimal subtree of $T_i$ that connects $S_j$. The embedding $P_i$ of this subtree corresponds to a graph $G_{ij}$ which is a subgraph of $G$. After eliminating isolated vertices, $G_{ij}$ is a connected subgraph of $G$. Let $T_{ij}$ be an arbitrary spanning tree of $G_{ij}$. We know that $T_{ij}$ spans $S_j$. We add $T_{ij}$ to the multi Steiner tree packing with weight $\lambda_i$.

This gives a multi Steiner tree packing solution that serves one unit demand, as the summation of $\lambda_i$ equals one. Denote by $\congestion_{uv}\in \mathbb{R}_{\ge0}$ the \emph{congestion} of edge $(u,v)$, which is the total amount of capacity of edge $(u,v)$ used by the packing solution over the capacity of edge $(u,v)$. The congestion $\phi$ of the packing solution is the maximum congestion over all edges. The multi Steiner tree number is then $\tau = \frac1{\phi}$ as it is the maximum fraction demand that could be satisfied while preserving the capacity constraint. For our multi Steiner tree packing solution, the congestion on edge $u,v$ is

\[
    \congestion_{uv} \le \frac{\sum_i \lambda_i\sum_{(x,y)\in T_i: (u,v)\in P_i(x,y)}D_i(x,y)}{c_{uv}} ~.
\]

Suppose we use $\alpha$-approximate cut-tree packing for the solution, we can upper bound the congestion,

\[
    \congestion_{uv} \le \alpha\cdot \frac{\sum_i \lambda_i\sum_{(x,y)\in T_i: (u,v)\in P_i(x,y)}D_i(x,y)}{\sum_i \lambda_i\sum_{(x,y)\in T_i: (u,v)\in P_i(x,y)}C_i(x,y)}\le \alpha \max_{i,x,y} \frac{D_i(x,y)}{C_i(x,y)} \le \alpha/\Psi ~.
\]

Therefore, the multi Steiner tree number $\tau = \frac{1}{\phi}\ge \frac1{\alpha}\cdot \Psi$. \Cref{lem: main LP int gap} is thereby a consequence of the following lemma. \qed

\begin{lemma}[\cite{racke2008optimal,williamson2011design}]
\label{thm: cut tree packing}
    There exists an $O(\log n)$-approximate cut-tree packing.
\end{lemma}

As discussed earlier, combining \Cref{lem: general sparsity ub coding} and \Cref{lem: main LP int gap} yields \Cref{thm: upper bound coding gap}. \qed

\section{Locally Decodable Codes in Network Coding}

In this section, we briefly discuss the limitations encountered when applying existing LDCs to network coding. We then introduce a novel variant of LDCs with a relaxed fault-tolerence property, specifically tailored to construct network coding gaps, so that we believe this variant is easier to construct than standard LDCs while still applicable to network coding.

\subsection{Limitations of Existing LDC Constructions}

\label{sec: ldc limit}

LDCs serve as the primary technical tool in this work and have been extensively studied for decades. Despite significant research, the fundamental trade-off between rate and query complexity remains a major open problem. For constant query complexity $q \ge 3$, the best known lower bound on the codeword length is $\Omega\left(k^{1+2/(q-2)}\right)$~\cite{basu2024improved}. In contrast, known upper bounds do not even reach the quasi-polynomial regime when $q = O((\log k)^c)$ for constant $c < 1$. Establishing an upper bound in this regime would directly imply network coding gaps of $(\log n)^{\Omega(1)}$ via our main lemma, highlighting the particular interest of the logarithmic-query regime.

The best known constructions of LDCs in the polylogarithmic-query regime are Matching Vector (MV) codes and Reed–Muller (RM) codes; their parameters are summarized in~\Cref{tab:comparison-table}, based on~\cite{yekhanin2012locally}. These codes satisfy perfect smoothness and tolerate a constant fraction of errors; that is, both $\delta$ and $\epsilon$ are constants. From the table, the rate of Reed–Muller (RM) codes lies in the quasipolynomial regime; however, this is achieved only when the query complexity is super-logarithmic, which is not the desired query regime. Both MV and RM codes have been studied for decades and improving them in either rate or query regime could be very challenging. 


\begin{table}[htbp]
    \centering
    \renewcommand{\arraystretch}{1.6}
    \begin{tabular}{|c|c|c|}
        \hline
        \textbf{Code Type} & $\boldsymbol{q}$ & $\boldsymbol{N}$ \\ 
        \hline
        MV code & $O(\log k)$ & 
        $\displaystyle \exp\left(\exp\left((\log k)^{O\left(\frac{\log\log q}{\log q}\right)} (\log\log k)^{\,1-\Omega\left(\frac{\log\log q}{\log q}\right)}\log q\right)\right)$ \\[8pt]
        \hline
        RM or MV codes & $O(\log k \log\log k)$ & 
        $\displaystyle k^{\,O(\log\log k)}$ \\[6pt]
        \hline
        RM code & $(\log k)^t,\; t>1$ & 
        $\displaystyle k^{\,1+\frac{1}{t-1}+o(1)}$ \\[6pt]
        \hline
    \end{tabular}
    \caption{Comparison of parameters for known LDC constructions.}
    \label{tab:comparison-table}
\end{table}

\subsection{Robust Distance LDC}
\label{sec: robust dist ldc}

In this section, we propose a novel variant of LDCs that is tailored for applications in network coding gaps. Recall from \Cref{sec: sampling ldcs} that the algorithm samples each sink with probability $(\log k)/k$, and consequently discards it with high probability. As a result, the full fault tolerance offered by traditional LDCs is not fully leveraged. This observation motivates us to define a new variant with relaxed fault tolerance requirements, which we believe is easier to construct than standard LDCs.

We begin with some preliminaries. A $q$-\emph{uniform hypergraph} $G = (V, E)$ is a hypergraph in which every edge $e \in E$ has rank $|e| = q$. A \emph{hypergraph matching} in $G$ is a subset of edges $H' \subseteq E$ such that no two edges overlap; that is, for all $e, e' \in H'$ with $e \ne e'$, we have $e \cap e' = \emptyset$. We introduce the notion of \emph{matching decodability}, a relaxation of local decodability.

\begin{definition}
\label{def: matching decodable}
A code $\C:\F^k \to \F^N$ is called \emph{$(q, \delta)$-matching decodable} if there exist $q$-uniform hypergraph matchings $H_1, \ldots, H_k$ on vertices indexed by $[N]$, each containing at least $\delta N$ hyperedges, such that for every $D \in H_i$, there exists a decoding function $f_{i,D}: \F^q \to \F$ satisfying
\[
\forall x \in \F^k,\quad x_i = f_{i,D}(\C(x)_v \mid v \in D).
\]
i.e., the $i$-th bit of the message can be decoded using the codewords whose indices lie in~$D$.
\end{definition}

This is a variant of the notion known as \emph{normally decodable}, which has been used in proving LDC lower bounds~\cite{yekhanin2012locally}.

Let $H = \bigcup_{i=1}^k H_i$. For any two hyperedges $e, e' \in H$, define the distance $\dist_H(e, e')$ as the length of the shortest path $e = e_1,e_2,\ldots,e_\ell = e'$ in $H$ connecting $e$ to $e'$, where consecutive hyperedges in the path must intersect (i.e., $\forall\, 1 \le i < \ell,\, e_i \cap e_{i+1} \ne \emptyset$).

\begin{definition}[Robust Distance LDC]
\label{def: rd ldc}
A $(q,\delta)$-matching decodable code is said to be \emph{$(q, \delta, d)$-robust distance locally decodable} if for each $i$, $|H_i| \ge \delta N$, and the following distance condition holds:
\[
\forall\, 1 \le i \le k,\ \forall\, e, e' \in H_i,\ e \ne e',\quad \dist_H(e, e') \ge d.
\]
\end{definition}

The advantage of this relaxed variant is highlighted by \Cref{lem: coding gap via rd}: the required matching size $\delta N$ can be significantly smaller—specifically, $\delta = \Omega((\log k)/k)$ suffices for our purposes, compared to the stricter requirement $\delta = \omega(1/\log k)$ for traditional LDCs to be useful for our network coding instances. On the other hand, we require the distance $d$ to be at least $\Omega((\log k)^{\epsilon})$ for some constant $\epsilon > 0$, to achieve a polylogarithmic coding gap. It is worth noting that, due to purely combinatorial limitations, $d$ can be at most $O(\log k / \log(\delta k))$. Thus, the code has to be sparse to obtain a large distance $d$. In particular, we show in \Cref{lem: ldc to rd} that standard LDCs can be transformed into robust distance LDCs using the same sampling technique introduced earlier. This demonstrates that robust distance LDCs are a relaxation of standard LDCs.

\begin{proposition}
\label{lem: coding gap via rd}
    If a $(q, \delta, d)$-robust distance locally decodable code $\C: \F^k\to \F^N$ exists, then the multi-source multicast network coding gap is at least $\Omega(\min\{\delta k, d/q\})$ for a graph of size $O(\delta Nk)$.
\end{proposition}

The proofs for both propositions are deferred to \Cref{app: rdldc-nc}, as it closely follows the argument in \Cref{sec: sampling ldcs}. 

\begin{proposition}
\label{lem: ldc to rd}
    Suppose there exists $(q, \delta, q\delta)$-locally decodable code with perfect smoothness, then a $\left(q, \Theta\bigl(\frac{\delta \log k}{qk}\bigr), \frac{\log k}{2(\log q + \log \log k)}\right)$-robust distance locally decodable code also exists.
\end{proposition}

\bibliographystyle{alpha}
\bibliography{ref}

\appendix

\section{Definition of Network Coding via Time-Expanded Graph}
\label{app: network coding}

In this section, we define the general network coding solution via the time-expanded graph construction, following prior work \cite{ahlswede2000network,harvey2006capacity}. For each undirected edge $e \in E$, we treat it as a pair of directed edges, $\vec{e}$ and $\backvec{e}$ (i.e., one in each direction). We will define their respective capacities shortly, ensuring that their sum equals the original capacity of $e$. From this point forward, we treat all edges as directed.

\begin{definition}[Time-Expanded Graph $G^*$] Given a directed graph $G = (V, E)$, the \emph{time-expanded graph} $G^* = (V^*, E^*)$ is a directed acyclic graph constructed as follows:
    \begin{itemize}
        \item The vertex set is $V^* = V \times \Z$. \item For each edge $e = (u, v) \in E$ and each time step $t \in \Z$, add a directed edge $e_t = (u^{(t-1)}, v^{(t)})$ to $E^*$.
        \item For each vertex $v \in V$ and each $t \in \Z$, add a \emph{memory edge} $(v^{(t-1)}, v^{(t)})$ to $E^*$.
    \end{itemize}
\end{definition}

We also model sources and sinks as special edges. For each source $s_i$, we introduce two auxiliary vertices $\sigma_i$ and $\sigma'_i$, and add edges $(\sigma_i, \sigma'_i)$ and $(\sigma'_i, s_i^{(t)})$ for all $t \in \Z$, each with infinite capacity. Similarly, for each sink $r \in R_i$, we add two auxiliary vertices $\gamma_{r,i}$ and $\gamma'_{r,i}$, and edges $(\gamma_{r,i}, \gamma'_{r,i})$ and $(r^{(t)}, \gamma_{r,i})$ for all $t \in \Z$, also with infinite capacity.

The alphabet on the edges $(\sigma_i, \sigma'_i)$ and $(\gamma_{r,i}, \gamma'_{r,i})$ is $M_i$. We say that the network coding solution satisfies \emph{correctness} if, for each sink $r \in R_i$, the symbol transmitted on $(\gamma_{r,i}, \gamma'_{r,i})$, denoted $X_{(\gamma_{r,i}, \gamma'_{r,i})}$, is equal to the message $m_i$ sent by the corresponding source, i.e., $X_{(\gamma_{r,i}, \gamma'_{r,i})} = X_{(\sigma_i, \sigma'_i)} = m_i$.

We now define the capacity constraints in the time-expanded graph. All memory edges have infinite capacity. For each directed edge $\vec{e}$, let the overall alphabet be $\Gamma(\vec{e}) = \prod_{t \in \Z} \Gamma(\vec{e}^{(t)})$, and let $X_{\vec{e}} \in \Gamma(\vec{e})$ denote the collection of symbols transmitted along all time steps of $\vec{e}$. The network coding solution respects the capacity constraint if: $H(\vec{e}) + H(\,\backvec{e}) \le c_e$, where by $H(\vec{e})$ we denote the entropy of the symbol $X_{\vec{e}}$ sended along $\vec{e}$.

\section{Linear Programs for Multi Steiner Tree Packing}
\label{app: Steiner packing}

The primal LP is as follows. Recall that for each demand $i$, the terminal set is $S_i = \{s_i\} \cup R_i$. Each tree $T$ in the Steiner tree packing for the $i$-th demand must span all nodes in $S_i$. The combined packing across all demands must respect the edge capacity constraints.

\begin{align*}
    \text{maximize} \quad & \tau \\
    \text{subject to} \quad & \sum_{T : S_i \subseteq V(T)} x_{i,T} \ge \tau \cdot d_i && \forall\, 1 \le i \le k \\
    & \sum_{\substack{i, T \\ e \in T}} x_{i,T} \le c_e && \forall\, e \in E \\
    & x_{i,T} \ge 0 && \forall\, i, T
\end{align*}

The dual of this linear program is as follows:

\begin{align*}
    \text{minimize} \quad & \frac{\sum_{e\in E} c_e y_e }{ \sum_i d_i z_i} \\
    \text{subject to} \quad & \sum_{e\in T} y_e \ge z_i && \forall\, i,T,\text{s.t. } S_i\subseteq V(T) \\
    & z_i \ge 0 && \forall\, 1\le i\le k \\
    & y_e \ge 0 && \forall\, e\in E
\end{align*}

In \Cref{sec: ub}, we mentioned that the generalized sparsity~$\Psi$ defined there serves as an upper bound on the integer program of the dual under the integral constraints $z_i, y_e \in \{0,1\}$ for all $i \in [k]$ and $e \in E$.
In fact, one can verify that these two quantities differ by at most a constant factor of two, although this observation is not used in our proof.
\section{Missing Proofs of \Cref{sec: robust dist ldc}}
\label{app: rdldc-nc}

\subsection{Proof of \Cref{lem: coding gap via rd}}

The gap instance $(G,\ses,F)$ is generated via the \Cref{alg: gap instance LDC via rd ldc}, where $G = (A,B,C,E)$ is a tripartite graph. All demands and edge capacities are units.

\begin{algorithm}[H]
\caption{Gap Instance via Robust Distance LDCs}
\label{alg: gap instance LDC via rd ldc}
\begin{algorithmic}[1]
\State Set $A \gets \{S\}$ and $s_i \gets S$ for all $i \in [k]$.
\State Set $B \gets \{w_i\}_{i=1}^N$.
\State Set $F \gets \{(S, w_i) \mid i \in [N]\}$ and initialize $E \gets F$.
\State Initialize $R_i \gets \emptyset$ for all $i \in [k]$.
\For{\textbf{each} $i=1, 2, \ldots, k$} 
    \For{\textbf{each} $D\in H_i$}
        \State Add new vertex $v$ to $R_i$.
        \State Add edges $(w_j,v)$ to $E$ for each $j\in D$.
    \EndFor
\EndFor
\State Set $C \gets \bigcup_{i=1}^k R_i$.
\end{algorithmic}
\end{algorithm}

The proof of network coding throughput is at least $a = 1$ is the same as \Cref{lem: sampling nc gap}.
    
For the other parameters in \Cref{lem: gap instance} for the gap instance, the size of the third layer dominates the number of vertices, and hence $n = \Theta(\delta Nk)$ by the robust distance LDC, as we assume $\delta = \Omega(1/k)$. The number of edges $m = \Theta(nq)$, and the number of sinks $r = \Theta(n)$. The distance $b$ is at least $\Omega(d)$, as for any $i\in [k]$ and any pair of decoding hyperedges in $H_i$, their distance is at least $d$ by \Cref{def: rd ldc}, and any path from a vertex $u\in R_i$ to another $v\in R_i$ in $G\setminus F$ in the gap instance maps to a path in the hypergraph $H_i$, such that the consecutive hyperedges in the path intersect. Therefore, the distance parameter is guaranteed by the robust distance LDC. The number of cut edges is $f = N$. Upon \Cref{lem: gap instance}, the coding gap is at least \[
        \frac{a\cdot r}{f+2m/b} = \Omega\left(\min\left\{\frac{r}{f},\frac{r}{2m/b}\right\}\right) = \Omega\left(\min\left\{\delta k, d/q\right\}\right). 
    \] \qed

\subsection{Proof of \Cref{lem: ldc to rd}}

The proof closely follows the argument in \Cref{sec: sampling ldcs}, which we now reinterpret its implication in the language of coding theory. We first construct the tripartite graph as in \Cref{alg: gap instance LDC}, with the only modification lying in the pruning subroutine: as we no longer assume \( q = o(\log k) \), we modify the pruning condition for a sink vertex at line~\ref{line: pruning cond} in \Cref{alg: gap instance LDC} to
\[
\text{dist}_{G \setminus F}(u, v) \le \frac{\log k}{\log q + \log\log k}. \
\]
With this new condition, the arguments in \Cref{lem: prune LDC instance} continue to hold under the updated parameters.

For each sink \( v \in R_i \), let \( D \) denote the coordinates of the codeword corresponding to the neighbors of \( v \) in the second layer \( B \). We add a hyperedge containing the coordinates in \( D \) to the hypergraph matching \( H_i \). We now show the third parameter of the robust distance LDC holds. The distance between any two hyperedges in \( H \) is at least half the distance between the corresponding sinks \( u \) and \( v \) in the tripartite graph \( G\setminus F \). Therefore, for any $i\in [k]$ and $e,e'\in H_i$, the distance of them in $H$ is then lower bounded by $\frac12\cdot \frac{\log k}{\log q + \log\log k}$.

The remaining is to lower bound the size of each matching $H_i$. At sampling subroutine, in total at least $\delta N/q$ hyperedges are generated and each is sampled to keep with probability $(\log k)/k$. The pruning subroutine then delete subconstant fraction of the matching. Therefore, the size $|H_i| = \Omega(\frac{\delta N \log k}{q\log\log k})$, implies the second parameter of the robust distance LDC code.

\section{Re-Analysis of Matching Vector Codes}

\label{app: mv code}

The main purpose of this section is to prove \Cref{lem: main mv code}.

\begin{definition}[Matching Vector Family]
    Let $S\subseteq \Z_m\setminus \{0\}$. The families of vertors $\U = \{u_i\}_{i=1}^n$ and $\V = \{v_i\}_{i=1}^n$ of vectors in $\Z_m^h$ is said to be \emph{$S$-matching family} if the following conditions hold:
    \begin{enumerate}
        \item $\langle u_i,v_i\rangle = 0$ for every $i\in [n]$.
        \item $\langle u_i,v_j\rangle \in S$ for every $i\ne j$.
    \end{enumerate}
\end{definition}

We will show how to use matching vector families to construct LDCs. We start with some preliminaries.

\begin{fact}[\cite{efremenko20093}]
    For every odd $m$ there exists $t\le m$ such that $m\mid 2^t-1$. For the finite field $\F_{2^t} = GF(2^t)$, there exists $g\in \F_{2^t}$ that generates a subgroup of size $m$, i.e. $g^m = 1$ and $g^i\ne 1$ for $1\le i < m$. 
\end{fact}

\begin{definition}
    Let $g$ be the generator of $\F$. A polynomial $P\in \F[x]$ is called an \emph{$S$-decoding polynomial} if the following conditions hold:
    \begin{itemize}
        \item $\forall i\in S, P(g^i) = 0$,
        \item $P(g^0) = P(1) = 1$.
    \end{itemize}
\end{definition}

\begin{fact}[\cite{efremenko20093}]
\label{fact: S decode polynomial}
    For any $S$ such that $0\not\in S$ there exists an $S$-decoding polynomial $P$ of degree $|S|$.
\end{fact}

\begin{lemma}[\cite{efremenko20093}]
\label{lem: decode}
    Let $\U,\V$ be a family of $S$-matching vectors in $\Z_m^h$, $|\U|=|\V|=k, |S|=s$. For $1\le t< m$ such that $m\mid 2^t - 1$, there exists a linear code $C:\F_{2^t}^k\to \F_{2^t}^{m^h}$ that is $(s+1,\delta,(s+1)\delta)$-locally decodable for all $\delta$, and is correct on uncorrupted codewords.
\end{lemma}

\begin{proof}
    We specify the encoding and decoding procedures. We simply denote the field $\F_{2^t}$ as $\F$.
    \begin{itemize}
        \item \textbf{Encoding:} Let $e_i\in \F^k$ be the $i$-th unit vector. We shall define $C(e_i) : \F^k\to \F^{m^h}$ for all $i$, and take the encoding function $C(\sum x_ie_i) := \sum x_iC(e_i)$ where the addition and multiplication are coordinate-wise. The encoding of $e_i$ is defined by a complete evaluation of a function $f_i: (\Z_m)^h\to \F$, which is defined as
        \[
            f_i(z) = g^{\langle u_i, z\rangle}, \text{ and } C_i(e_i) = (f_i(z))_{z\in (\Z_m)^h} ~.
        \]
        \item \textbf{Decoding:} The input to the decoder is a (corrupted) codeword $y$ and an index $i$. Write the $S$-decoding polynomial in \Cref{fact: S decode polynomial} as $P(z) = a_0 + a_1 z + a_2 z^2 + \ldots + a_s z^s$. The decoder performs the following:
        \begin{enumerate}
            \item Pick $w\in \Z_m^h$ uniformly at random.
            \item Query the codeword at $y(w),y(w+v_i),y(w+2v_i),\ldots,y(w+s\cdot v_i)$.
            \item Output $g^{-\langle u_i, w\rangle} (a_0 y(w) + a_1 y(w+v_i) + a_2 y(w+2v_i) + \ldots + a_s y(w+s\cdot v_i))$.
        \end{enumerate}
    \end{itemize}
    We show that the code is correct on uncorrupted codewords and smoothness, implying the lemma by \Cref{fact: smooth code imply ldc}. The smoothness is due to the uniform random of the point $w$.

    To show correctness on uncorrupted codewords, for any $i$ and $w$,
    \begin{align*}
        a_0 y(w) + a_1 y(w+v_i) + a_2 y(w+2v_i) + \ldots + a_s y(w+s\cdot v_i) = \sum_{\ell=0}^s a_\ell \sum_{j=1}^k x_j g^{\langle u_j,w+\ell\cdot v_i \rangle} ~.
    \end{align*}
    Rearranging terms, we have
    \begin{align*}
        \sum_{\ell=0}^s \sum_{j=1}^k x_j g^{\langle u_j,w+\ell\cdot v_i \rangle} & = \sum_{j=1}^k g^{\langle u_j,w\rangle} x_j \sum_{\ell=0}^s a_\ell \left(g^{\langle u_j,v_i \rangle}\right)^{\ell} \\
        & = \sum_{j=1}^k g^{\langle u_j,w\rangle} x_j P\left(g^{\langle u_j,v_i \rangle}\right) \\
        & = g^{\langle u_i,w\rangle} x_i ~. \qedhere
    \end{align*}
\end{proof}

\begin{definition}
    Let $m = \prod_{i=1}^t p_i$ be a product of distinct primes. The \emph{canonical set} in $\Z_m$ is the set of all non-zero $s$ such that for every $i\in [t]$, $s\in \{0,1\}\bmod p_i$.
\end{definition}

\begin{lemma}[\cite{dvir2011matching}]
\label{lem: mv family}
    Let $m =\ $\scalebox{0.9}{$\displaystyle\prod_{i=1}^t$}$~p_i$ be a product of distinct primes. Let $w$ be a positive integer. Let $\{e_i\},i\in[t]$ be integers such that for all $i$, we have $p_i^{e_i}>w^{1/t}$. Let $\displaystyle d=\max_ip_i^{e_i}$ and $h\ge w$ be arbitrary. Let $S$ be the canonical set; then there exists an \scalebox{1.2}{${\binom{h}{w}}$}-sized family of $S$-matching vectors in $\Z_m^n$, where $n =\ $\scalebox{1.2}{$\binom{h}{\le d}$}.
\end{lemma}

\begin{proof}[Proof of \Cref{lem: main mv code}]
The proof follows by setting appropriate parameters in \Cref{lem: mv family} and applying \Cref{lem: decode}.

\begin{enumerate}
    \item By a strengthened version of Bertrand’s postulate (Theorem 5.8 in \cite{shoup2009computational}), there exists a constant \( c \) such that the interval \([ (c/2)t \ln t, ct \ln t ]\) contains at least \( t \) distinct odd primes \( p_1, \ldots, p_t \).
    
    \item Let \( m = \prod_{i \in [t]} p_i \). This implies \( m = t^{\Theta(t)} \).

    \item The size of canonical set is $s = |S| = 2^t - 1$.
    
    \item Define \( b \) as the smallest positive integer such that \( m \mid 2^b - 1 \). Then \( b \le m - 1 \).
    
    \item Assume \( w \) is a multiple of \( t \), and set \( k = w^{w/t} \). Then it follows that \( w = \Theta\left(\frac{t \log k}{\log \log k}\right) \).
    
    \item Let \( d = \max_i p_i^{e_i} = O(w^{1/t} \cdot t \ln t) \).
    
    \item Set \( h = c' \cdot w^{1 + 1/t} \), where \( c' \) is a sufficiently large constant to ensure \( h \ge d \).
    
    \item Observe that
    \[
    \binom{h}{w} \ge \left(\frac{h}{w}\right)^w \ge k.
    \]
    
    \item Also, note that
    \[
    \binom{h}{\le d} \le d \cdot \left(\frac{e h}{d}\right)^d.
    \]
    
    \item Then the total code length is
    \[
    N = m^{\binom{h}{\le d}} \le \exp\exp\left(t \ln t \cdot w^{1/t} \ln w\right) \le \exp\exp\left(O\left((\log k)^{1/t} (\log\log k)^{1 - 1/t} \cdot t \ln t\right)\right).
    \]
\end{enumerate}

Finally, the assumption in (5) \( k = w^{w/t} \) can be made without loss of generality. If \( k \) does not exactly have this form, we can pad the message with zeros to obtain a message of length \( k' \) that satisfies this structure. This padding incurs at most a quadratic blowup, which does not affect the asymptotic bounds.

\end{proof}

\end{document}